\documentclass{amsart}

\usepackage{amssymb, amsmath, amsthm, amsfonts, amscd}

\usepackage[colorlinks=true]{hyperref}

 \usepackage[msc-links]{amsrefs}

\newtheorem{theorem}{Theorem}
\newtheorem{lem}{Lemma}
\newtheorem{cor}{Corollary}
\newtheorem{prop}{Proposition}
\newtheorem{defn}{Definition}
\newtheorem{example}{Example}

\newcommand{\res}[2]{\mathrm{res}_{#1} \left( #2 \right)}

\newcommand{\supp}{{\rm supp}}

\def\g{\gamma}
\newcommand{\X}{{\mathcal X}}
\newcommand{\M}{\mathcal M}
\def\C{\mathbb C}
\def\F{\mathbb F}
\def\Aut{\mbox{Aut }}
\def\PAut{\mbox{PAut }}
\def\MAut{\mbox{MAut }}
\def\GAut{\Gamma\mbox{Aut }}
\def\iso{\cong}
\def\a{\alpha}
\def\P{\mathbb P}
\def\<{\langle}
\def\>{\rangle}
\def\s{\sigma}
\def\N{\mathcal N}
\def\H{\mathcal H}
\def\L{\mathcal L}
\def\x{\textbf{x}}
\newcommand\Z{\mathbb Z}
\def\P{\mathbb P}

\title{On the automorphism groups of some AG-codes based on $C_{a, b}$ curves}

\author{T. Shaska}
\address{Department of Mathematics and Statistics, Oakland University, Rochester, MI, 48309-4485, USA.\\
}

\author{H. Wang}
\address{Department of Mathematics, School of Science, Beihang University, Beijing 100083, P. R. China. }
\keywords{$C_{a, b}$ curves, AG-codes, automorphism groups, superelliptic curves}

\subjclass[2010]{Primary: 11T71, 14G50,  Secondary: 94A05, 14Q05}

\begin{document}

\maketitle


\begin{abstract}
We study $C_{a, b}$ curves and their applications to coding theory.    We show how $C_{a, b}$ curves can be
used to construct MDS codes and focus on some $C_{a, b}$ curves with extra automorphisms, namely $y^3=x^4+1,
y^3=x^4-x, y^3-y=x^4$. The automorphism groups of such codes are determined in most characteristics.
\end{abstract}


\section{Introduction}

In the design of new codes algebraic geometry codes (AG-codes), also known as Goppa codes, play an important part
and have been well studied in the last few decades. In designing such codes an important fact is the number of
points of the algebraic curve over a finite field. Hence, it is natural that the algebraic curves that have been
used so far are curves for which such number of points can be computed. Is there a "nice" family of curves which
can be used to construct good codes? Hermitian curves have been used successfully by many authors in addition to
hyperelliptic curves and other families of curves. The most natural curves are superelliptic curves as shown in \cite{e-sh, i-sh, shor-sh,  sh-38, book }

 In this paper we investigate a   family of curves which belongs to superelliptic curves, namely the $C_{a, b}$ curves. $C_{ab}$ curves are algebraic curves with very
interesting arithmetic properties. There are algorithms suggested which count the number of points of these curves
using the Monsky-Washnitzer cohomology; see \cite{Denef}. In this paper we study how these properties can be used
in constructing good algebraic geometry codes.

In Section 2,   we give a brief introduction to algebraic geometry codes (AG-codes).  This is well-known material. As a standard reference we use \cite{coding-book}.  Many other excellent resources exist.

In Section 3, we briefly define $C_{ab}$ curves. Such
curves have degree $a, b$ covers to $\P^1$. The existence of certain divisors makes these curves useful in coding
theory. For more details on AG-codes and quantum AG-codes one can check \cite{sh-38, e-sh, sh-25, sh-35}.  

In Section 4, we study the locus of genus $g$, $C_{a, b}$ curves  for fixed $a, b$. Such curves have degree $a, b$
covers to $\P^1$. Such covers are classified according to the ramification structure. We assume that the cover
has the largest possible moduli dimension. This determines a ramification structure $\s$. The simplest case of $C_{a, b}$ curves are hyperelliptic curves and superelliptic curves.  Such curves have been studied in detail by many authors and are well understood. They are the main classes of curves being used in coding theory and cryptography. 

Denote the moduli spaces of these maximal moduli dimension degree $a, b$ coverings by $\M_a$ and $\M_b$
respectively and let $g = \frac 1 2 (a-1) (b-1)$. Then $\M_a, \M_b$ are algebraic varieties of $\M_g$ (not
necessarily irreducible). The locus of $C_{a, b}$ curves in $\M_g$ is the intersection $\M_a \cap \M_b$.
Studying this locus is the focus of section 3.

In the last section we use $C_{a, b}$ curves of genus 3 to construct AG-codes. Such codes are MDS codes. We
focus on some genus 3 $C_{ab}$ curves with extra automorphisms, namely $y^3=x^4+1, y^3=x^4-x, y^3-y=x^4$. The
automorphism groups of such codes are determined for some characteristics.

This is an updated version of a small note from 2006. Connections to superelliptic curves are added and an updated list of references. 

\bigskip

\noindent \textbf{Notation:} Throughout this paper $\X$ will denote a smooth, projective curve defined over some
field $F$. By $\Aut (\X)$ we denote the group of automorphisms of $\X$ defined over $\bar F$. By $C$ we will
denote a linear code. The permutation automorphism group of $C$ will be denoted by $\PAut (C)$, the monomial
automorphism group by $\MAut (C)$, and the automorphism group by $\GAut (C)$. $\F_q$ denotes a finite field of $q$
elements.

\section{Preliminaries}
Let $\F_q$ be a finite field of size $q$ and $\X$ a genus $g\geq 2$  algebraic curve defined over $\F_q$.  Let 
 $F$ be the  function field of $\X$ and  $P_1, \ldots, P_n$ be points of multiplicity on in $\X$

We take divisors   $D = P_1 + \cdots + P_n$  and  $G$ such that  $\supp(G) \cap \supp(D) = \emptyset$. In addition $\L (G)$ denotes the Riemann-Roch space for the divisor $G$. The \textit{algebraic geometry code}     $C_\L \subseteq \F_q^n$ is defined by
\begin{eqnarray*}
    C_\L(D,G) = \left\{ (f(P_1), \ldots , f(P_n)) \; | \; f \in \L(G) \right\} \subseteq \F_q^n
\end{eqnarray*}
The following linear map is called the  \textit{evaluation map}
  \begin{eqnarray*}
    \varphi: & &
    \left\{
    \begin{array}{crl}
      \L(G) & \rightarrow & \F_q^n\\
      f & \mapsto & (f(P_1), \ldots, f(P_n)).
    \end{array}
    \right.
  \end{eqnarray*}
Thus, the code is given by 
\[C_\L(D,G) = \varphi (\L(G)).\]   
It is a linear code  $[n, k, d]$   with parameters
\[
\begin{split}
    k & =  \dim G - \dim (G-D), \\
     d  & \ge  n - \deg G . 
\end{split}
\]
%
The following result is  well known, see \cite[Thm. II.2.3]{st} among many others.

\begin{lem}
If $\deg G < n$, then
\begin{enumerate}
  \item $ \varphi: \L(G) \rightarrow C_\L(D,G)$ is injective and
    $C_\L(D,G)$ is an $[n,k,d]$ code with
\[       k  =  \dim G \ge \deg G + 1 - g,  \]
and 
\[  d  \ge  n - \deg G. \]
  \item If in addition $2g-2 < \deg G < n$, then 
  \[ k  = \deg G + 1 - g.\]
  \item If $(f_1, \ldots, f_k)$ is a basis of $\L(G)$, then
    \begin{eqnarray*}
      M & = &
      \left(
      \begin{array}{ccc}
        f_1(P_1) & \cdots & f_1(P_n)\\
        \vdots   &        & \vdots\\
        f_k(P_1) & \cdots & f_k(P_n)\\
      \end{array}
      \right)
    \end{eqnarray*}
    is a {\it generator matrix} for $C_{\mathcal L}(D,G)$.
  \end{enumerate}
\end{lem}

To characterize the dual code of a AG-code we need to look at the original definitions of Goppa by means of
differential forms and its relations to the code defined above. We define the code $C_{\Omega}(D,G)$ by
\begin{eqnarray*}
C_{\Omega}(D,G)  :=  \{ (\res{P_1}{\omega}, \ldots, \res{P_n}{\omega}) \; | \; \omega \in \Omega_F(G-D)  \}
\subseteq \F_q^n.
\end{eqnarray*}

The code $C_{\Omega}(D,G)$, where $D$ and $G$ are as above is equal to the dual of $C_\L(D,G)^{\perp}$.  In other words, 
$C_\L(D,G)^{\perp} = C_{\Omega}(D,G)$.
Also,
\[ C_{\Omega}(D,G) = a \cdot C_\L(D,H)\]
with $H = D - G + (\eta)$    where $\eta$ is a differential, $v_{P_i}(\eta) = -1$ for 
$i = 1, \ldots, n$, and $a = \left( \res{P_1}{\eta}, \ldots, \res{P_n}{\eta} \right)$.
Moreover,   
\[C_{\mathcal L}(D,G)^{\perp} = a \cdot C_\L(D,H).\]

The following proposition is cited from \cite[Prop. VII.1.2]{st}. It allows to construct differentials with
special properties that help to construct a self-orthogonal code.

\begin{lem} \label{constrdiff}
Let $x$ and $y$ be elements of $F$ such that $v_{P_i}(y) = 1$, $v_{P_i}(x) = 0$ and $x(P_i) = 1$ for $i = 1,
\ldots, n$. Then the differential $\eta := x \cdot \frac{dy}{y}$ satisfies
\[v_{P_i}(\eta) = -1, \quad  and \quad \res{P_i}{\eta} = 1\]
for $i = 1, \ldots, n$.
\end{lem}

Next we give some standard definitions on automorphism groups of codes, which will be the main focus of this paper.  A basic reference is \cite{coding-book}.

The  \textit{permutation automorphism group} of the code $C \subseteq F_q^n$ is the subgroup of $S_n$ (acting on
$F_q^n$ by coordinate permutation) which preserves $C$. We denote such group by $\mbox{PAut } (C)$. 

The set of
monomial matrices that map $C$ to itself forms the \textit{monomial automorphism group}, denoted by $\MAut (C)$.
Every monomial matrix $M$ can be written as $M=DP$ where $D$ is a diagonal matrix and $P$ a permutation matrix.
Let $\g$ be a field automorphism of $\F_q$ and $M$ a monomial matrix. Denote by $M_\g$ the map $M_\g : C \to C$
such that $\forall \x \in C$ we have $ M_\g (\x) = \g (M\x)$.
The set of all maps $M_\g$ forms the \textit{automorphism group} of $C$, denoted by $\GAut (C)$. It is well known
that
\[ \PAut (C) \leq \MAut (C) \leq \GAut (C) \]

Next we will define as admissible a class of curves which have some additional conditions on their divisors.

\begin{defn}
\label{admiss} A genus $g \geq 2$ curve $\X/F_q$ is called \textit{admissible} if it satisfies the following properties:

i) there exists a rational point $P_{\infty}$ and two functions $x,~y\in F$ such that
\[ (x)_{\infty}=kP_{\infty},~(y)_{\infty}=lP_{\infty},\]
 and $k,~l\geq 1$;

ii)  for $m\geq0$, the elements $x^iy^j$ with $0\leq i, 0\leq j\leq k-1$, and $ki+lj\leq m$ form a basis of the
space $\L(mP_{\infty})$.
\end{defn}

\noindent  Next we define
\[Aut_{D,G}(\X) : =  \left\{\sigma\in Aut(\X) | \,   \sigma (D)=D~ \mbox{and}~ \sigma (G)=G  \right\}.\]
With the above notation  we have the following:

\begin{theorem} \label{iso}
Let $\X/F_q$  be an admissible curve over $F_q$ of genus $g$ where $l>k$. Assume that $m\geq l$. Let $D=\sum_{P\in
J}P$ where $J\subseteq \mathbb{P}\backslash \{P_{\infty}\}$, $\mathbb{P}$ is the set of all rational points of
$\X$. If
\[n>  \max    \left\{2g+2, \, 2m, \, k \left(l+\frac{k-1}{\beta} \right),  \, lk \left(1+\frac{k-1}{m-k+1} \right)   \right\},\]
where  $ \ n=|J|$ and 
\[ \beta= \min  \{k-1,~r \, | \, y^r\in \L(mP_{\infty})\}, \]
then
\[
\Aut(C_\L(D, mP_{\infty}))\iso \Aut_{D, mP_{\infty}}(\X).
\]
\end{theorem}

\begin{proof} See  \cite{Wes}  for details \end{proof}

In the next two sections we will see how we can compute the automorphism groups of certain AG-codes constructed by some superelliptic curves.

\section{Introduction to $C_{ab}$ curves}
In this section, we introduce the notion of $C_{ab}$ curves which constitute a wide class of algebraic curves
including elliptic curves, hyperelliptic curves and superelliptic curves. They have been studied by many people
due to their nice properties.

Throughout this section $k$ denotes an algebraically closed field of characteristic not equal to 2. 
Let $a$ and $b$ be relatively prime positive integers. Then a curve $\X$ defined over $k$ is called an $C_{ab}$ curve  if it is a nonsingular plane curve defined by $f(x, y ) = 0$, where $f(x, y) \in k[x, y]$ has the form
\begin{equation}\label{c_ab}
f(x, y ) = \alpha_{0,a} y^a + \alpha_{b,0} x^b + \sum_{ai+bj<ab} \alpha_{i,j} \, x^i y^j  
\end{equation}
for nonzero $\alpha_{0,a},~\alpha_{b,0}\in k$.
%

Let $\X$ be a $C_{ab}$ curve defined over $k$. There exists exactly one $k$-rational place $\infty$ at infinity, which implies that the degree of  $\infty$ is $1$. Furthermore, the pole divisors of $x$ and $y$ are $a\cdot\infty$ and $b\cdot\infty$, respectively. The genus of $\X$ is 
\[ g (\X) = \frac {(a-1)(b-1)} 2.\]

Hence, $C_{a, b}$ curves have  fully ramified degree $a$ and $b$ covers to $\P^1 (k)$. 
Consider first the degree
$a$ cover $\pi_a: C_{a, b} \to \P^1 (k)$. Since the cover is fully ramified then there are at least $2g+a-1$
other branch points. Thus, the total number of branch points is
\[ d_1 := 2g+a= ab-b+1=b(a-1)+1\]
The degree $b$ cover has
\[ d_2:=2(g-1)+2b -(b-1) = (a-1)(b-1)+b+1=a(b-1)+1\]
branch points.

\begin{cor}
All hyperelliptic curves are $C_{a, b}$ curves.
\end{cor}

\proof Every genus $g$ hyperelliptic curve can be written as $y^2=f(x)$ such that $\deg f = 2g+1$. Take $a=2$
and $b=2g+1$. 

\endproof

Next, we will see superelliptic curves which are even a larger class than that of hyperelliptic curves, but first the following example. 

\begin{example} Let $a=3, b=4$. Then the genus of the curve is $g=3$ and we have
\begin{small}
\begin{equation}
Y^3 + \a_1 X^4 + \a_2 X^3 + \a_3 X^2Y + \a_4 X Y^2 + \a_5 X^2 + \a_6 Y^2 + \a_7 XY + \a_8 X + \a_9 Y + \a_{10} =0
\end{equation}
\end{small}
\end{example}

Since the dimension of $\M_3$ is 5 then we should be able to write this curve in a "better" way; see the next
section for details. The next proposition will be useful when we construct AG-codes from $C_{ab}$ curves.

\begin{prop}\label{basis}
Let $\X$ be a $\ C_{ab}$  curve defined by $f(X, Y) = 0$ with $f(X, Y )\in F[X, Y ]$. Then
\[\{X^i Y^j \ | \   0 \leq j \leq a-1, \ i\geq 0, \  ai + bj \leq m \}\]
is a basis of a vector space $\L(m\cdot\infty)$ over $F$, where  $\ m \in Z_{\geq 0}$.
\end{prop}

\begin{cor}
$C_{a, b}$ curves are admissible curves.
\end{cor}

Hence we can use the results of the previous section when constructing codes from such curves.

\subsection{Superelliptic curves}

There are a special class of $C_{ab}$ curves which are well understood due to the work of many authors \cite{zhupa, b-e-sh, b-th, b-sh-sh, sh-37, bin,i-sh, 
sh-20,  rachel}.

Let $\X$ be a genus $g \geq 2$ algebraic curve defined over $k$, $G$ its automorphism group, and $H$ a subgroup of $G$ of order $|H|=m$, inside the center $Z (G)$, such that the genus of the quotient space $\X/H$ is zero.  Such curves are called \textit{superelliptic curves} and they can be written with the affine equation 
$y^m = f(x)$ for some $f \in k[x]$; see \cite{sh-37} for more details.  

The following lemma is an immediate consequence of the definition.
\begin{lem}
Superelliptic curves are $C_{ab}$ curves.
\end{lem}
Then we have the following. 
\begin{cor}
Superelliptic curves are admissible curves.
\end{cor}

In \cite{sa} are determined are possible groups of superelliptic curves defined over fields of characteristic $\neq 2$.  In \cite{sa-1} are determined even the equations for each group.  This is not known for algebraic curves in general. 

\subsection{Automorphism groups of $C_{ab}$ curves}
Let $\X$ be a $C_{ab}$ curve as above. Can we  determine the automorphism group of
$\X$ over $k$ in terms of $a, b$? For genus $g=2, 3$ such groups can be determined by work of
previous authors; see \cite{sh-8} for genus 2 curves and \cite{sa, sa-1} for genus 3 superelliptic curves.  In \cite{sh-15} is treated the case of genus 3 non-hyperelliptic curves. 
For higher genus such groups can be determined if the $C_{ab} $ curve is hyperelliptic or superelliptic.  In general there is no known algorithm to determine the automorphism group of an algebraic curve. 

\begin{lem}Let $\X$ be a genus $g=2$ algebraic curve as in Eq.~\eqref{c_ab} defined over $k$.
Then  $\Aut (\X)$ is isomorphic to one of the following:

i) $p=3$:  $ \Z_2, V_4, D_4, D_6, GL_2(3),$

ii) $p= 5$: $\Z_2, \Z_{10}, V_4, D_4, D_6,   GL_2(3), $.

iii) $p\geq 5$: $\Z_2, V_4, D_4, D_6, SL_2(3)$.
\end{lem}

For the case $p=2$ see \cite{sh-8} for details. For $g=3$ see \cite{sa, sa-1}.  The automorphism groups of superelliptic curves defined over a field $k$ such that $\mbox{char} k \neq 2$ are determined completely in \cite{sa, sa-1} and the corresponding equations are determined in \cite{sh-27}. 

\section{The locus of $C_{3, 4}$ curves in the moduli space $\M_3$}
In this section we want to focus on non-hyperelliptic genus 3 curves. More precisely, we want to study the
space of $C_{3, 4}$ curves in the moduli space $\M_3$.  Throughout this section all curves are defined over a
characteristic zero field.

We first give a brief introduction to the Hurwitz spaces and projection of such spaces on $\M_g$.  Let $X$ be
a curve of genus $g$ and $f: X \to \P^1$ be a covering of degree $n$ with $r$ branch points. We denote the
branch points by $q_1, \ldots, q_r \in \P^1$ and let $p\in \P^1\setminus\{q_1,\ldots,q_r\}$. Choose loops
$\gamma_i$ around $q_i$ such that
\[\Gamma:=\pi_1 (\P^1\setminus\{ q_1, \dots , q_r\},\ p)=\< \gamma_1, \ldots , \gamma_r\>, \quad \gamma_1 \cdots \gamma_r=1.\]
$\Gamma$ acts on the fiber $f^{-1}(p)$ by path lifting, inducing a transitive subgroup $G$ of the symmetric group
$S_n$ (determined by $f$ up to conjugacy in $S_n$). It is called the \emph{monodromy group} of $f$. The images of
$\gamma_1,\ldots,\gamma_r$ in $S_n$ form a tuple of permutations $\s=(\s_1,\ldots,\s_r)$ called a tuple of
\emph{branch cycles} of $f$. We call such a tuple the \emph{signature} of $\phi$. The covering $f:X\to\P^1$ is of
type $\s$ if it has $\s$ as tuple of branch cycles relative to some homotopy basis of $\P^1\setminus \{ q_1,
\dots , q_r\}$.

Two coverings $f:X\to\P^1$ and $f':X'\to\P^1$ are \emph{weakly equivalent} (resp. \emph{equivalent}) if there is
a homeomorphism $h:X\to X'$ and an analytic automorphism $g$ of $\P^1$ such that $g\circ f=f'\circ h$ (resp.,
$g=1$). Such classes are denoted by $[f]_w$ (resp., $[f]$). The \emph{Hurwitz space} $\H_\s$ (resp.,
\emph{symmetrized Hurwitz space $\H_\s^s$}) is the set of weak equivalence classes (resp., equivalence) of covers
of type $\s$, it carries a natural structure of an quasiprojective variety.

Let $C_i$ denote the conjugacy class of $\s_i$ in $G$ and $C=(C_1, \dots , C_r)$. The set of Nielsen classes
$\N(G, C)$ is
\[\N(G,\s):=\{(\s_1,\dots,\s_r)\ |\ \s_i\in C_i,\,G=\<\s_1,\dots,\s_r\>,\ \s_1\cdots\s_r=1\}\]
The braid group acts on $\N(G, C)$ as
\[[\g_i]: \quad (\s_1, \dots , \s_r) \to (\s_1, \, \dots , \, \s_{i-1}, \s_{i+1}^{ \s_i}, \s_i, \s_{i+2}, \dots, \s_r)\]
where $\s_{i+1}^{ \s_i}= \s_i \s_{i+1} \s_i^{-1}$. We have $\H\sigma=\H_\tau$ if and only if the tuples $\s$,
$\tau$ are in the same \emph{braid orbit} $\mathcal O_\tau = \mathcal O_\sigma$.

Let $\M_g$ be the moduli space of genus $g$ curves. We have morphisms
\begin{equation}
\begin{split}
& \H_\s \overset {\Phi_\s} \longrightarrow \H_\s^s \overset { \bar{\Phi}_\s }  \longrightarrow \M_g\\
 & [f]_w                       \to [f] \to [X]
\end{split}
\end{equation}
Each component of $\H_\s$ has the same image in $\M_g$. We denote by
\[ \L_g :={\bar \Phi}_\s (\H_\s^s).\]
We say that the covering $f$ or the ramification $\s$ has \emph{moduli dimension} $\delta:=dim \L_g$.

Let $a, b$ be fixed and $g= (a-1)(b-1)/2$. The generic $C_{a, b}$ curve of genus $g$ has a degree $a$ cover
$\pi_a : C_{a, b} \to \P^1$ (resp. degree $b$ cover $\pi_b :  C_{a, b} \to \P^1$).

The ramification structure of $\pi_a : C_{a, b} \to \P^1$ is $(a, 2, \dots , 2)$ where the number of branch
points is $d_1=b(a-1)+1$, as discussed in section 3. Let $\H_a$ denote the Hurwitz space of such covers and
$\M_a$ its image in $\M_g$, as described above. Then, the dimension of $\M_a$ ia $\delta_1 \leq b(a-1)-2$.
Similarly, we get that the dimension of $\M_b$ is $\delta_2 \leq a(b-1)-2$. Of course, the cover with
smallest degree among $\pi_a$ and $\pi_b$ is the one of interest. From now on, we assume that $a< b$.

As mentioned above the goal of this section is to study the space $\M_{a, b}$ for fixed $a$ and $b$,
particularly on the case $a=3$ and $b=4$.

\begin{theorem}\label{thm_var} Every genus 3 curve is a $C_{3,4}$ curve. Moreover, every genus $C_{3, 4}$
curve  defined over a field $k$,   is isomorphic   to a   curve with equation
\begin{equation}\label{eq} f(x,y)=(x+b)y^3+(cx+d)y^2+(ex^2+fx)y+x^3+kx^2+lx = 0. \end{equation}
\end{theorem}

\proof The case of hyperelliptic curves is obvious. Hence, we focus on non-hyperelliptic genus 3 curves. Let
$C$ be a non-hyperelliptic genus 3 curve, $P$ be a Weierstrass point on $C$, and $K$ the function field of
$C$. Then exists a meromorphic function $x$ which has $P$ as a triple pole and no other poles. Thus, $[K :
L(x)]=3$. Consider $x$ as a mapping of $C$ to the Riemann sphere. We call this mapping $\psi: C \to \P^1$
and let $\infty$ be $\psi(P)$. From the Riemann-Hurwitz formula we have that $\psi$ has at most 8 other
branch points. There is also a meromorphic function $y$ which has $P$ as a pole of order 4 and no other
poles. Thus the equation of $K$ is given by
\begin{equation} \label{eq_g3}
f(x,y):=\g_1(x)\, y^3 + \g_2(x)\, y^2 + \g_3(x)\, y +\g_0 (x)=0
\end{equation}
where $\g_0 (x), \dots , \g_3 (x) \in L[x]$ and
\[ deg(\g_0)=4, \ \deg (\g_1)=0,  \  \deg (\g_2)\leq 2, \ \deg (\g_3)\leq 3. \]
The discriminant of $f(x,y)$ with respect to $y$
$$D(f,y):=-27\, (\g_1 \, \g_0)^2 + 18\, \g_0\, \g_1\, \g_2\, \g_3 + (\g_2\, \g_3)^2 - 4\, \g_0 \, \g_2^3
- 4\, \g_1\, \g_3^3, $$
must have at most degree 8 since its roots are the branch points of $\psi:\C \to \P^1$. 

Thus, we have
\[
\begin{split}
& \deg \, (\g_3 \g_2) \leq 4, \\
& \deg\, (\g_0 \g_2^3) \leq 8, \\
&  \deg \, (\g_3^3\, \g_1)\leq 8. 
\end{split}
\]
If $deg\, (\g_2)=2$ then $deg \, (\g_3)\leq 2$ and $deg \, (\g_0)=0$. Thus, $deg\, (f,x)=2$. Then, $f(x,y)=0$
is not the equation of an genus 3 curve. Hence, $deg \, (\g_2) \leq 1$. Clearly, $deg\, (\g_3)\leq 1$. We
denote:
\begin{equation}
\begin{split} 
\g_1 (x):= & \, a, \\
 \g_2 (x):= &\, cx+d\\
\g_3 (x):= & \, ex+f, \\
 \g_0 (x):= & \,gx^4+ hx^3+kx^2+lx+m\\
\end{split}
\end{equation}
Then, we have
\[ f(x,y)=ay^3+(cx+d)y^2+(ex+f)y+ (g x^4+h x^3+k x^2+l x+m)=0\]
which is obviously an $C_{3, 4}$ curve. This completes the proof of the first statement.

Let $C$ be a $C_{3,4}$ curve defined over $k$. Then, $C$ is a non-hyperelliptic genus 3 curve. Hence, it is
isomorphic over $k$ to a curve with equation as in Eq.~\eqref{eq}; see \cite{sh-15} for details. This
completes the proof.

\qed

Hence, the space of $C_{3, 4}$ curves correspond to the moduli space $\M_3$. It is an interesting problem to
see what happens for higher genus $g$.

\section{Codes obtained from $C_{a,b}$ curves}
In this section we will give examples of codes which are constructed based on $C_{ab}$ curves.  We will focus
on three curves, namely $y^3= x^4+1$, $y^3 = x^4-x$, and $y^3-y= x^4$. All these curves are genus 3
non-hyperelliptic curves. For characteristic $p > 7$ these curves have  automorphism group   isomorphic to a
group with GAP identity (48, 33), (9,1), and (96, 64) respectively; see \cite{sh-3} for details. Recall that an
$[n,k, d]$ code with $d=n-k+1$ is called \textit{maximum distance separable} code or an MDS code.

\subsection{The curve $y^3 = x^4+1$}
Let $\X$ be the curve \[ y^3 = x^4+1\] defined over $\F_q$. This is a $C_{3,4}$ curve of genus $3$.  For
characteristic $p \neq 2,3 $ the  automorphism group of $\X$ is $C_4 \rtimes A_4$, which has Gap identity
$(48, 33)$. We denote the set of affine rational points of $\X$ over $\F_q$ by $\{P_1, \dots , P_n\}$. Let
$C= C_\L(D,G)$, where $n+1$ is the number of rational points of $\X$ and
\[ G = m P_\infty, \ \ D= P_1 + \cdots P_n \]
We have the following result:
\begin{theorem}
For the permutation automorphism group $\PAut (C)$,    one has

i)   If  $\ 0 \leq  m < 3$ or $   m > n+4$ then $\PAut (C) \iso S_n$.

ii)   If $n>24$ and $4\leq m < n/2 $ then $\PAut (C) \iso Aut_{D, m P_\infty}(\X)$.
\end{theorem}

\begin{proof}
i)  If $0 \leq  m < 3$, then from Proposition \ref{basis} we know $(1, 1, \cdots, 1)$ is a basis of the vector
space $\L(m\cdot P_\infty)$ , thus $\dim G=1$. Since $\dim (G-D)\geq 0,~\dim C\geq 1$, together with $\dim C=
\dim G - \dim (G-D)$ we have  $\dim C=1$. Therefore $\PAut (C) \iso S_n$.

If $   m > n+4$, then $deg(G-D)>2g-2$. Thus $\dim C= \dim G - \dim (G-D)=n.$ $C$ is the full space, therefore
$\PAut (C) \iso S_n$.

ii) With the notation of definition \ref{admiss}. In this case $k=3,~l=4,~g=3,~\beta=1,~m\geq 4$. For Theorem
\ref{iso} to hold we need
\[n>   \max   \left\{8, 2m, 18, 12  \left(1+\frac{2}{m-2} \right)      \right\}.\]

Since $m\geq 4$, $12(1+\frac{2}{m-2})\leq
24$. Thus when $n>24$ and $4\leq m < n/2 $ Theorem \ref{iso} applies and $\PAut (C) \iso Aut_{D, m P_\infty}(\X)$.

\end{proof}

It can be seen from the proof of theorem that $C_\L(D,G)$ is a $[n, 1, n]$ MDS code when $0 \leq  m < 3$, and a $[n, n, 1]$ MDS code when  $   m > n+4$.

\begin{example}
Let $\X$ be defined over $F_{2^3}$. Take $m=4$. By computation using GAP, we find that $C_\L(D,G)$ is a
$[8, 3, 6]$   MDS code with a generator matrix $$ \left(\begin{array}{cccccccc}
         \alpha^5 & \alpha^3& \alpha^6& 1 & \alpha^4 & \alpha & \alpha^2& 0\\
           \alpha^3 & \alpha^6& \alpha^5& 0 & \alpha^2 & \alpha^4 & \alpha& 1\\
           1 & 1 & 1& 1 & 1  & 1& 1 & 1\\
        \end{array}\right),$$ where $\alpha$ is a primitive
        element of $F_{2^3}$. The permutation automorphism group is $\PAut (C_\L(D,G) ) \iso Z_{14}$.
\end{example}

\subsection{The curve $y^3 = x^4-x$}
Let $\X$ be the curve \[ y^3 = x^4-x\] defined over $\F_q$.  For characteristic $p>7$ the automorphism group
of $\X$ is the cyclic group of order 9.  Denote the set of affine rational points of $\X$ over $\F_q$ by
$\{P_1, \dots , P_n\}$. Let $C= C_\L(D,G)$, where $n+1$ is the number of rational points of $\X$ and
\[ G = m P_\infty, \ \ D= P_1 + \cdots P_n \]
We have the following result:
\begin{theorem}
For the permutation automorphism group $\PAut (C)$,    one has

i)  If  $0 \leq  m < 3$ or $   m > n+4$ then $\PAut (C) \iso S_n$.

ii) If $n>24$ and $4\leq m < n/2 $ then $\PAut (C) \iso Aut_{D, m P_\infty}(\X)$.

\end{theorem}

\begin{proof}
i)  If $0 \leq  m < 3$, then from Proposition \ref{basis} we know $(1, 1, \cdots, 1)$ is a basis of the vector
space $\L(m\cdot P_\infty)$ , thus $\dim G=1$. Since $\dim (G-D)\geq 0,~\dim C\geq 1$, together with $\dim C=
\dim G - \dim (G-D)$ we have  $\dim C=1$. Therefore $\PAut (C) \iso S_n$.

If $   m > n+4$, then $deg(G-D)>2g-2$. Thus $\dim C= \dim G - \dim (G-D)=n.$ $C$ is the full space, therefore
$\PAut (C) \iso S_n$.

ii) With the notation of definition \ref{admiss}. In this case $k=3,~l=4,~g=3,~\beta=1,~m\geq 4$. For Theorem
\ref{iso} to hold we need
\[
n > \max \left\{    8, 2m, 18, 12 \left(   1 + \frac{2} {m-2} \right)  \right\}
\] 
Since $m\geq 4$ and  $12  \left(1+\frac{2}{m-2} \right) \leq 24$. Thus when $n>24$ and $4\leq m < n/2 $ Theorem \ref{iso} applies and $\PAut (C) \iso Aut_{D, m P_\infty}(\X)$.
\end{proof}
 
It can be seen from the proof of theorem that $C_\L(D,G)$ is a $[n, 1, n]$ MDS code when $0 \leq  m < 3$, and a $[n, n, 1]$ MDS
code when  $   m > n+4$.

\begin{example}
Let $\X$ be defined over $F_{2^3}$. Take $m=3$. By computation using GAP, we find that $C_\L(D,G)$ is a $[8,
2, 7]$ code with permutation automorphism group $[56, 11]$(Gap identity), which is clearly an MDS code.
\end{example}

\subsection{The curve $y^3-y= x^4$}
Let $\X$ be the curve \[ y^3-y= x^4\] defined over $\F_q$. Denote the set of affine rational points of $\X$ over
$\F_q$ by $\{P_1, \dots , P_n\}$. Let $C=C_\L(D,G)$, where $n+1$ is the number of rational points of $\X$ and
\[ G = m P_\infty, \ \ D= P_1 + \cdots P_n \]

We have the following;
\begin{theorem}
For the permutation automorphism group $\PAut (C)$,    one has

i)  If  $0 \leq  m < 3$ or $   m > n+4$ then $\PAut (C) \iso S_n$.

ii) If $n>24$ and $4\leq m < n/2 $ then $\PAut (C) \iso Aut_{D, m P_\infty}(\X)$.
\end{theorem}

\begin{proof}
i)  If $0 \leq  m < 3$, then from Proposition \ref{basis} we know $(1, 1, \cdots, 1)$ is a basis of the vector
space $\L(m\cdot P_\infty)$ , thus $\dim G=1$. Since $\dim (G-D)\geq 0,~\dim C\geq 1$, together with $\dim C=
\dim G - \dim (G-D)$ we have  $\dim C=1$. Therefore $\PAut (C) \iso S_n$.

If $   m > n+4$, then $deg(G-D)>2g-2$. Thus $\dim C= \dim G - \dim (G-D)=n.$ $C$ is the full space, therefore
$\PAut (C) \iso S_n$.

ii) With the notation of Definition \ref{admiss}. In this case $k=3,~l=4,~g=3,~\beta=1,~m\geq 4$. For Theorem
\ref{iso} to hold we need
\[
n> \max \left\{8, 2m, 18, 12(1+\frac{2}{m-2}) \right\}.
\] 
Since $m\geq 4$, $12(1+\frac{2}{m-2})\leq
24$. Thus when $n>24$ and $4\leq m < n/2 $ Theorem \ref{iso} applies and $\PAut (C) \iso Aut_{D, m
P_\infty}(\X)$.

\end{proof}

It can be seen from the proof of theorem that $C_\L(D,G)$ is a $[n, 1, n]$ MDS code when $0 \leq  m < 3$,
and a $[n, n, 1]$ MDS code when  $   m > n+4$.

\begin{example}
Let $\X$ be defined over $F_{2^2}$. Take $m=6$. By computation using GAP, we find that $C_\L(D,G)$ is a
$[4, 4, 1]$ code with a generator matrix
$$ \left(\begin{array}{cccc}
         \alpha & \alpha^2& 0& 0\\
           \alpha^2 & \alpha& 0 & 0\\
        \alpha & \alpha^2& 1& 0\\
           1 & 1 & 1& 1 \\
        \end{array}\right),$$ where $\alpha$ is a primitive
        element of $F_{2^2}$.
The permutation automorphism group is isomorphic to the group with GAP identity $[24, 12]$. In this case
\[ \PAut(C)  \hookrightarrow    \Aut (\X). \]
This code is clearly an MDS code.
\end{example}

\section{Concluding remarks}
It is an open question to determine the automorphism groups of AG-codes obtained by $C_{a, b}$ curves, in all
characteristics. Moreover, even determining the list of automorphism groups of $C_{a, b}$ curves seems to be
a non-trivial problem.

Furthermore, to determine the locus of $C_{a, b}$  curves, defined over $\C$, in the moduli space $\M_g$,
where $g = \frac {(a-1)(b-1)} 2$, seems an interesting problem in its own right. A $C_{a, b}$ curve has
covers of degree $a$ and $b$ to $\P^1$. One has to take such covers in the most generic form. The space of
$C_{a, b}$ curves in $\M_g$ will be the intersection of the corresponding Hurwitz spaces. To generalize this
for any $a, b$ would require a careful analysis of the corresponding Hurwitz spaces.

Since the first version of this note, considerable progress has been made with superelliptic curves.  Such curves are well understood and their automorphism groups fully determined in all characteristics different from two, due to work of Sanjeewa \cite{sa}. Moreover, for all such groups we can determine the equation of the corresponding curve \cite{sa-1}.  For curve with extra automorphisms such equations can be determined over the minimal field of definition due to work of Beshaj/Thompson \cite{b-th} and Hidaldo/Shaska \cite{h-sh}.  Furthermore, due to work of Beshaj such equations over the minimal field of definition can even be chosen with minimal coefficients \cite{bin}.

It is still unknown whether a precise relation exists between the automorphism group of the curve and the automorphism group of the Ag-codes, even in the case of superelliptic curves whose automorphism groups are well understood. 


\bibliographystyle{amsplain} 

\bibliography{arxiv}{}

\begin{bibdiv}
\begin{biblist}

\bib{b-e-sh}{article}{
      author={Beshaj, L.},
      author={Elezi, A.},
      author={Shaska, T.},
       title={Theta functions of superelliptic curves},
        date={201503},
      eprint={1503.00297},
         url={http://arxiv.org/abs/1503.00297},
}

\bib{bin}{article}{
      author={Beshaj, Lubjana},
       title={Reduction theory of binary forms},
        date={201502},
      eprint={1502.06289},
         url={http://arxiv.org/abs/1502.06289},
}

\bib{sh-37}{article}{
      author={Beshaj, Lubjana},
      author={Hoxha, Valmira},
      author={Shaska, Tony},
       title={On superelliptic curves of level {$n$} and their quotients, {I}},
        date={2011},
        ISSN={1930-1235},
     journal={Albanian J. Math.},
      volume={5},
      number={3},
       pages={115\ndash 137},
      review={\MR{2846162 (2012i:14036)}},
}

\bib{b-sh-sh}{article}{
      author={Beshaj, Lubjana},
      author={Shaska, Tony},
      author={Shor, Caleb},
       title={On {J}acobians of curves with superelliptic components},
        date={201310},
     journal={Contemporary Mathematics, Volume 629, 2014, pg. 1-14},
      eprint={1310.7241},
         url={http://arxiv.org/abs/1310.7241},
}

\bib{book}{book}{
      editor={{Beshaj}, Lubjana},
      editor={{Shaska}, Tony},
      editor={{Zhupa}, Eustrat},
       title={{Advances on superelliptic curves and their applications. Based
  on the NATO Advanced Study Institute (ASI), Ohrid, Macedonia, 2014.}},
    language={English},
   publisher={Amsterdam: IOS Press},
        date={2015},
        ISBN={978-1-61499-519-7/hbk; 978-1-61499-520-3/ebook},
}

\bib{b-th}{article}{
      author={Beshaj, Lubjana},
      author={Thompson, Fred},
       title={Equations for superelliptic curves over their minimal field of
  definition},
        date={201405},
     journal={Albanian J. Math. Vol. 8 (2014), no. 1, 3-8},
      eprint={1405.4556},
         url={http://arxiv.org/abs/1405.4556},
}

\bib{Denef}{article}{
      author={Denef, Jan},
      author={Vercauteren, Frederik},
       title={Counting points on {$C_{ab}$} curves using {M}onsky-{W}ashnitzer
  cohomology},
        date={2006},
        ISSN={1071-5797},
     journal={Finite Fields Appl.},
      volume={12},
      number={1},
       pages={78\ndash 102},
         url={http://dx.doi.org/10.1016/j.ffa.2005.01.003},
      review={\MR{2190188 (2007c:11075)}},
}

\bib{e-sh}{article}{
      author={Elezi, A.},
      author={Shaska, T.},
       title={Quantum codes from superelliptic curves},
        date={201305},
     journal={Albanian J. Math. 5 (2011), no. 4, 175--191},
      eprint={1305.3941},
         url={http://arxiv.org/abs/1305.3941},
}

\bib{coding-book}{book}{
      author={Huffman, W.~Cary},
      author={Pless, Vera},
       title={Fundamentals of error-correcting codes},
   publisher={Cambridge University Press, Cambridge},
        date={2003},
        ISBN={0-521-78280-5},
         url={http://dx.doi.org/10.1017/CBO9780511807077},
      review={\MR{1996953}},
}

\bib{h-sh}{inproceedings}{
      author={Hidalgo, R.},
      author={Shaska, T.},
       title={On the field of moduli of superelliptic curves},
        date={2016},
   booktitle={Algebraic curves and their fibrations in mathematical physics and
  arithmetic geometry},
}

\bib{i-sh}{article}{
      author={Izquierdo, Milagros},
      author={Shaska, Tony},
       title={Cyclic curves over the reals},
        date={201501},
      eprint={1501.01559},
         url={http://arxiv.org/abs/1501.01559},
}

\bib{sh-3}{article}{
      author={Magaard, K.},
      author={Shaska, T.},
      author={Shpectorov, S.},
      author={V{{\"o}}lklein, H.},
       title={The locus of curves with prescribed automorphism group},
        date={2002},
     journal={S\=urikaisekikenky\=usho K\=oky\=uroku},
      number={1267},
       pages={112\ndash 141},
        note={Communications in arithmetic fundamental groups (Kyoto,
  1999/2001)},
      review={\MR{1954371}},
}

\bib{sh-20}{article}{
      author={Previato, E.},
      author={Shaska, T.},
      author={Wijesiri, G.~S.},
       title={Thetanulls of cyclic curves of small genus},
        date={2007},
        ISSN={1930-1235},
     journal={Albanian J. Math.},
      volume={1},
      number={4},
       pages={253\ndash 270},
      review={\MR{2367218 (2008k:14066)}},
}

\bib{sa}{article}{
      author={Sanjeewa, R.},
       title={Automorphism groups of cyclic curves defined over finite fields
  of any characteristics},
        date={201301},
     journal={Albanian J. Math. Vol. 3, Number 4, 2009, 131-160},
      eprint={1301.4593},
         url={http://arxiv.org/abs/1301.4593},
}

\bib{sh-35}{book}{
      editor={Shaska, T.},
      editor={Hasimaj, E.},
       title={Algebraic aspects of digital communications},
      series={NATO Science for Peace and Security Series D: Information and
  Communication Security},
   publisher={IOS Press, Amsterdam},
        date={2009},
      volume={24},
        ISBN={978-1-60750-019-3},
        note={Papers from the Conference ``New Challenges in Digital
  Communications'' held at the University of Vlora, Vlora, April 27--May 9,
  2008},
      review={\MR{2605610 (2011a:94002)}},
}

\bib{sh-38}{article}{
      author={Shaska, Tony},
       title={Quantum codes from algebraic curves with automorphisms},
        date={2008},
     journal={Condensed Matter Physics},
      volume={11},
      number={2},
       pages={383\ndash 396},
}

\bib{sh-27}{article}{
      author={Shaska, T.},
       title={Subvarieties of the hyperelliptic moduli determined by group
  actions},
        date={201302},
     journal={Serdica Math. Journal, No. 4, 355-374, 2006},
      eprint={1302.3974},
         url={http://arxiv.org/abs/1302.3974},
}

\bib{rachel}{article}{
      author={Shaska, Rachel},
       title={Equations of curves with minimal discriminant},
        date={201407},
      eprint={1407.7064},
         url={http://arxiv.org/abs/1407.7064},
}

\bib{sh-25}{book}{
      editor={Shaska, T.},
      editor={Huffman, W.~C.},
      editor={Joyner, D.},
      editor={Ustimenko, V.},
       title={Advances in coding theory and cryptography},
      series={Series on Coding Theory and Cryptology},
   publisher={World Scientific Publishing Co. Pte. Ltd., Hackensack, NJ},
        date={2007},
      volume={3},
        ISBN={978-981-270-701-7; 981-270-701-8},
         url={http://dx.doi.org/10.1142/9789812772022},
        note={Papers from the Conference on Coding Theory and Cryptography held
  in Vlora, May 26--27, 2007 and from the Conference on Applications of
  Computer Algebra held at Oakland University, Rochester, MI, July 19--22,
  2007},
      review={\MR{2435341 (2009h:94158)}},
}

\bib{sa-1}{article}{
      author={Sanjeewa, R.},
      author={Shaska, T.},
       title={Determining equations of families of cyclic curves},
        date={201301},
     journal={Albanian J. Math. VOL 2, NO 3 (2008)},
      eprint={1301.4591},
         url={http://arxiv.org/abs/1301.4591},
}

\bib{shor-sh}{article}{
      author={Shaska, T.},
      author={Shor, C.},
       title={Weierstrass points of superelliptic curves},
        date={201502},
      eprint={1502.06285},
         url={http://arxiv.org/abs/1502.06285},
}

\bib{sh-15}{incollection}{
      author={Shaska, T.},
      author={Thompson, J.},
       title={On the generic curve of genus 3},
        date={2005},
   booktitle={Affine algebraic geometry},
      series={Contemp. Math.},
      volume={369},
   publisher={Amer. Math. Soc., Providence, RI},
       pages={233\ndash 243},
         url={http://dx.doi.org/10.1090/conm/369/06814},
      review={\MR{2126664 (2006c:14042)}},
}

\bib{st}{book}{
      author={Stichtenoth, Henning},
       title={Algebraic function fields and codes},
     edition={Second},
      series={Graduate Texts in Mathematics},
   publisher={Springer-Verlag, Berlin},
        date={2009},
      volume={254},
        ISBN={978-3-540-76877-7},
      review={\MR{2464941 (2010d:14034)}},
}

\bib{sh-8}{incollection}{
      author={Shaska, T.},
      author={V{{\"o}}lklein, H.},
       title={Elliptic subfields and automorphisms of genus 2 function fields},
        date={2004},
   booktitle={Algebra, arithmetic and geometry with applications ({W}est
  {L}afayette, {IN}, 2000)},
   publisher={Springer, Berlin},
       pages={703\ndash 723},
      review={\MR{2037120 (2004m:14047)}},
}

\bib{zhupa}{article}{
      author={Shaska, Tony},
      author={Zhupa, Eustrat},
      author={Beshaj, Lubjana},
       title={The case for superelliptic curves},
        date={201502},
      eprint={1502.07249},
         url={http://arxiv.org/abs/1502.07249},
}

\bib{Wes}{article}{
      author={Wesemeyer, Stephan},
       title={On the automorphism group of various {G}oppa codes},
        date={1998},
        ISSN={0018-9448},
     journal={IEEE Trans. Inform. Theory},
      volume={44},
      number={2},
       pages={630\ndash 643},
         url={http://dx.doi.org/10.1109/18.661509},
      review={\MR{1607734 (99m:94058)}},
}

\end{biblist}
\end{bibdiv}

\end{document}